\documentclass[conference]{IEEEtran}
\IEEEoverridecommandlockouts
\usepackage{amsmath,amssymb,amsthm}
\usepackage[ruled,vlined,linesnumbered]{algorithm2e}
\usepackage{tikz}
\usetikzlibrary{positioning}
\usepackage[caption=false,font=footnotesize]{subfig}
\usepackage{graphicx}
\usepackage{placeins}
\usepackage{dblfloatfix}
\usepackage{enumitem}
\usepackage{hyperref}
\usepackage{cite}
\usepackage{pbalance}

\newcommand{\ignore}[1]{}

\newtheorem{definition}{Definition}
\newtheorem{property}{Property}
\usepackage{graphicx}
\usepackage{pstricks}
\usepackage{times}

\newcommand{\NP}{\mathrm{NP}}
\newcommand{\APX}{\mathrm{APX}}
\renewcommand{\epsilon}{\varepsilon}
\usepackage{amsmath}
\usepackage{amssymb}
\usepackage{graphicx}
\usepackage{cite}
\usepackage{verbatim} 
\usepackage{url}
\usepackage{flushend}
\usepackage[english]{babel}
\usepackage{balance}
\newtheorem{theorem}{Theorem}
\newtheorem{lemma}{Lemma}
\newtheorem{corollary}{Corollary}

\tikzset{cell/.style={draw,rounded corners=2pt,inner sep=3pt,minimum width=5.5mm,minimum height=4.2mm,font=\footnotesize},
         strip/.style={draw,rounded corners=2pt,inner sep=2pt,minimum height=4mm,font=\scriptsize},
         rev/.style={red,thick}, bdry/.style={black!60,dashed}}

\begin{document}

\title{Sorting by Strip Swaps is NP-Hard}

\author{
\IEEEauthorblockN{
Swapnoneel Roy
}
\IEEEauthorblockA{\textit{School of Computing}\\
\textit{University of North Florida}\\
Jacksonville, Florida, USA\\
Email: s.roy@unf.edu}
\and
\IEEEauthorblockN{Asai Asaithambi}
\IEEEauthorblockA{\textit{Computer Science Department}\\
\textit{Florida Polytechnic University}\\
Lakeland, Florida, USA \\
Email: aasaithambi@floridapoly.edu}
\and
\IEEEauthorblockN{Debajyoti Mukhopadhyay}
\IEEEauthorblockA{\textit{WIDiCoReL Research Lab}\\
Mumbai - India \\
Orcid ID: 0000-0002-8079-4091\\
debajyoti.mukhopadhyay@gmail.com}
}
\maketitle

\begin{abstract}
We show that \emph{Sorting by Strip Swaps} (SbSS) is NP-hard by a polynomial reduction of \emph{Block Sorting}. The key idea is a local gadget, a \emph{cage}, that replaces every decreasing adjacency $(a_i,a_{i+1})$ by a guarded triple $a_i,m_i,a_{i+1}$ enclosed by guards $L_i,U_i$, so the only decreasing adjacencies are the two inside the cage. Small \emph{hinge} gadgets couple adjacent cages that share an element and enforce that a strip swap that removes exactly two adjacencies corresponds bijectively to a block move that removes exactly one decreasing adjacency in the source permutation. This yields a clean equivalence between exact SbSS schedules and perfect block schedules, establishing NP-hardness.
\end{abstract}

\section{Introduction}

\noindent
Sorting by restricted operations on permutations is a central theme with connections to comparative genomics and data rearrangement. In \emph{Block Sorting} one picks up a maximal increasing block and inserts it elsewhere. Block-sorting is NP-hard~\cite{bein2003block,mahajan2007block,narayanaswamy2015}.

\noindent
In \emph{Sorting by Strip Swaps} (SbSS) the allowed operation is to swap two (maximal) increasing strips. Although approximation algorithms are known, the definitive hardness proof requires care because a naive reduction may allow local fixes to create new decreases elsewhere~\cite{asaithambi2017implementation,huang2016new,roy2007algorithms,roy2008sorting}. We present a simple, schedule-free reduction that avoids this pitfall and proves NP-hardness.

When the genomes of two organisms are compared, it is assumed that the genomes consist of the same set of genes. \ignore{It is the order of the genes in theses genomes that makes them distinguishable.} In order to study the similarities between two genomes, the minimum number of rearrangements or mutations required to transform one genome into the other is a very important metric. In computer science, these genomes are represented as permutations, and the rearrangements are represented as \textit{primitives}. \ignore{or rules which are applied to determine the shortest distance between two genomes (permutations).}  Some well known primitives are reversals~\cite{berman20021}, transpositions~\cite{bafna1998sorting}, and block interchanges~\cite{christie1996sorting, lin2005efficient}. We are interested in transforming an arbitrary starting permutation to a target permutation, which is considered to be the {\em identity} permutation by applying the primitives to \ignore{\textit{arbitrary}} substrings within the starting permutation. \ignore{Hence the process corresponds to sorting a given permutation using the minimum number of steps of a primitive or a combination of primitives. This leads to various combinatorial optimization problems that are of interest in their own right. }

In this paper, we consider applying the primitives to maximal substrings in the starting permutation that are also substrings in the identity permutation, which are called \textit{blocks} \cite{bein2003block, mahajan2007block, bein2005faster, huang2016new, asaithambi2017implementation}, or \textit{strips} \cite{roy2007algorithms, roy2008approximate}. The motivation for defining a block or a strip in this manner is to emphasize that any substring that is already sorted will not be broken. Note, however, that the term block has been used by some researchers to refer to any substring in the starting permutation \cite{christie1996sorting, lin2005efficient, berman20021, bafna1998sorting}. Therefore, in order to avoid confusion, we will use the term \textit{strips} to refer to the already-sorted maximal substrings of the starting permutation.

\section{Sorting By Strip Swaps}
We begin with an example illustrating the ideas we have presented so far so that we may define the problem of sorting-by-strip-swaps more formally. Consider the starting permutation.
\[\ 2\ 5\ 6\ 3\ 7\ 8\ 9\ 4\ 1,\ \]
which consists of the six strips 2, 5 6, 3, 7 8 9, 4, and 1. For ease of presentation and understanding, we will display this strip structure as follows.
\begin{center}
\begin{tabular*}{1.8in}{@{\extracolsep{\fill}}cccccc}
\fbox{2}& \fbox{5\ 6}& \fbox{3}& \fbox{7\ 8\ 9}& \fbox{4}&  \fbox{1}
\end{tabular*}
\end{center}
Now, swapping the strip \fbox{5\ 6} with the strip \fbox{3} in the above permutation results in the permutation with the following strip structure.
\begin{center}
\begin{tabular*}{1.6in}{@{\extracolsep{\fill}}cccc}
\fbox{2\ 3}& \fbox{5\ 6\ 7\ 8\ 9}&  \fbox{4}&  \fbox{1}
\end{tabular*}
\end{center}
Note that the number of strips has been reduced to 4 from 6 after one strip swap.
Next, swapping the strip \fbox{5\ 6\ 7\ 8\ 9} with the strip \fbox{4} in the above permutation yields the permutation with the following strip structure.
\begin{center}
\begin{tabular*}{1.6in}{@{\extracolsep{\fill}}cc}
\fbox{2\ \  3\ \ 4\ \ 5\ \ 6\ \ 7\ \ 8\ \ 9}& \fbox{1}
\end{tabular*}
\end{center}
Note that the second swap has reduced the number of strips from 4 to 2. Finally, swapping the two strips in the above permutation yields the identity permutation (which is just one strip), and we show this as follows.
\begin{center}
\ignore{\begin{tabular*}{1.883in}{@{\extracolsep{\fill}}c}
\fbox{\ 1\, \,  2\, \ 3\, \ 4\ \  5\ \  6\, \, 7\, \, 8\, \, 9\ }\\
\end{tabular*}
}
\begin{tabular}{c}
\fbox{\ 1\, \,  2\, \ 3\, \ 4\ \  5\ \  6\, \, 7\, \, 8\, \, 9\ }\\
\end{tabular}
\end{center}
The sequence of strip swaps is known as a \textit{strip swap schedule}. The schedule in the above example is written as 
\[ \left(\ \fbox{5\ 6} \leftrightarrow \fbox{3}\ \right), 
   \left(\ \fbox{5\ 6\ 7\ 8\ 9} \leftrightarrow \fbox{4}\ \right),\]
 \[  \left(\ \fbox{2\ \  3\ \ 4\ \ 5\ \ 6\ \ 7\ \ 8\ \ 9} \leftrightarrow \fbox{1}\ \right).\]
Observe that we were able to obtain the identity permutation from the starting permutation in three strip-swaps. We are interested in answering the question ``Can we accomplish the above task in fewer than three strip swaps''? In relation to the genomic problem we described earlier, we would like to know if we could accomplish this with fewer strip swaps. Specifically, we are interested in transforming the starting permutation to the identity permutation in the smallest number of strip swaps, a combinatorial optimization problem. 

\newcommand{\SSD}{\textnormal{SSD}}
\newcommand{\SSS}{\textsc{Sorting by Strip Swaps}}

\begin{definition}[Strip Swap Distance] 
Let $\pi$ be a permutation of the $n$ integers from $[n]$. The strip swap distance $\SSD(\pi)$ from $\pi$ to the identity permutation $id_n$ is the minimum integer value $m$ for which there are $m$ strip swaps \ignore{$b_1, b_2, \cdots, b_m$} which when applied sequentially to $\pi$, produce $id_n$.   
\end{definition}
A decision version of this problem may be stated as follows:
\begin{center}
\framebox{
\begin{tabular}{l}
\noindent {\sc Sorting by Strip Swaps}
\\ \noindent {\sc Input:} A permutation $\pi$ and an integer $m$.
\\ \noindent {\sc Question:}  Is $\SSD(\pi) \leq m$?
\end{tabular}
}
\end{center}
The authors of \cite{christie1996sorting} prove that sorting by a minimal number of block interchanges, where a block may be any substring of the permutation, is polynomially solvable and present a $O(n^2)$ algorithm for solving this problem. On the other hand, although the strip-swap primitive may be viewed as a non-trivial variant of the block-interchange primitive \cite{roy2008sorting}, the minimal block-interchange algorithm to polynomially solve the sorting-by-block-interchange problem does not work for \SSS{}. However, approximation algorithms for \SSS{} have been designed in~\cite{roy2007algorithms} and~\cite{roy2008approximate}. Although the $2$-approximation algorithm in~\cite{roy2008approximate} is the best known algorithm to date, the computational hardness of \SSS{} has remained an open question.

Another importance of \SSS{} lies in its similarity with {\sc Block Sorting}. The term block in {\sc Block Sorting} refers to our term strip. \ignore{The {\sc Block Sorting} problem is motivated mainly by its application in optical character recognition~\cite{bein2003block}}. In this problem, we are allowed to pick a strip and place it anywhere in the permutation in each move. {\sc Block Sorting} corresponds to finding the minimum number of such strip moves required to sort $\pi$. {\sc Block Sorting} has been shown to be $\NP$-Hard~\cite{bein2003block}, $\APX$-Hard~\cite{narayanaswamy2015}, and the best known approximation algorithms are 2-approximation algorithms~\cite{mahajan2007block, bein2005faster}. \ignore{We note that a polynomial time algorithm for \SSS{} would imply a $2$-approximation algorithm for {\sc Block Sorting}, because a single strip swap move can be mimicked by $2$ strip moves. Hence the study of \SSS{} might lead us to the design of simpler approximation algorithms for {\sc Block Sorting}, a problem which has been open for a while now. Additionally, it is known that a block sorting move (or a strip move) reduces the number of strips by at most $3$ at each step. In contrast, a strip swap reduces the number of strips at most by $4$ at each step.  Since the maximum number of strips that can reduced by a strip swap is {\em more} than that of a block sorting move, algorithms for \SSS{} are expected to perform {\em better} than those for {\sc Block Sorting}. } A little insight into the problem of {\sc Block Sorting} will reveal the fact that a block sorting move in fact also {\em interchanges} the positions of {\em two} sub-strings in $\pi$: one is the block being moved, and the other is its adjacent substring, which might or might not be a block (strip). Hence the cost of a block move operation is essentially {\em equal} to the cost of a strip swap. Hence it is really interesting to know about the computational hardness of \SSS.

\noindent
\paragraph*{Notation.} We write permutations in one-line notation $\pi=a_1\ldots a_n$. A \emph{reversal boundary} is an index $i$ with $a_i>a_{i+1}$. Let $\mathrm{rev}(\pi)$ be the number of reversal boundaries. Basic lower bounds: for Block Sorting, $bs(\pi)\ge \mathrm{rev}(\pi)$; for SbSS, any swap fixes at most two boundaries, so $\mathrm{SSD}(\pi)\ge \lceil \mathrm{rev}(\pi)/2\rceil$.
\section{Overview of Results and Techniques}\label{over}
We prove the $\NP$-Hardness of \SSS{} by reducing $3SAT$ to it via {\sc Block Sorting}, which is the problem of sorting a given permutation using the minimum number of strip moves, where a {\em strip move} displaces a  strip to a different position. In order to accomplish this, we need the following concepts, results, and observations established previously by other researchers.

\newcommand{\rev}{\textnormal{rev}}

\noindent Let $\pi_x$ denote the position of the strip $x$ in the permutation $\pi$.

\begin{definition}[Reversal] 
In a permutation $\pi$, a reversal is a pair of consecutive elements $a$ and $b$ such that $a>b$. Formally $a$ and $b$ form a reversal in $\pi$ if $a>b$ and $\pi_b = \pi_a+1$.
\end{definition}
Let the number of reversals in $\pi$ be $\rev(\pi)$. In~\cite{bein2003block}, a block sorting sequence of length $\rev(\pi)$ has been shown to be optimal, since the block sorting distance $bs(\pi) \geq \rev(\pi)$. 
\begin{center}
    \begin{minipage}{0.5\textwidth}
    We show that $\SSD(\pi) \geq \rev(\pi)/2$. 
    \end{minipage}
\end{center} 
Hence, a strip-swap sequence of length $\rev(\pi)/2$ is optimal for $\pi$. Furthermore, note that since $\rev(\pi)/2$ is only a lower bound for \SSS, there might exist a permutation $\pi$ such that $\SSD(\pi)>\rev(\pi)/2$.

\noindent
Given an arbitrary permutation $\pi$, we construct another permutation $\pi^{\prime}$ from $\pi$ such that $bs(\pi)=\rev(\pi)$ if and only if $\SSD(\pi^{\prime})=\rev(\pi^{\prime})/2$.

\noindent
In~\cite{bein2003block} the authors had constructed the permutation $\pi$ from an arbitrary $3SAT$ formula $\Phi$ such that $\Phi$ is satisfiable if and only if $bs(\pi)=\rev(\pi)$. Then, in conjunction with what we plan to show about the permutation $\pi^{\prime}$ constructed from $\pi$, we would have the result.
\begin{center}
    \begin{minipage}{0.8\textwidth}
 $\Phi$ is satisfiable if and only if $\SSD(\pi^{\prime})=\rev(\pi^{\prime})/2$.
    \end{minipage}
\end{center} 
Since we will show that $\SSD(\pi) \geq \rev(\pi^{\prime})/2$, we also have the result 
\begin{center}
    \begin{minipage}{0.8\textwidth}
 $\Phi$ is satisfiable if and only if $\SSD(\pi^{\prime}) \leq \rev(\pi^{\prime})/2$.
    \end{minipage}
\end{center} 
The above proves that \SSS{} is $\NP$-Hard, which is our main result. In summary, if $f$ is the polynomial-time algorithm described in~\cite{bein2003block} to construct $\pi$ from a given $\Phi$, and $g$ is the polynomial-time algorithm we design to construct $\pi^{\prime}$ from $\pi$, then we have the following:
\noindent
\begin{center}
    \begin{align*}
    3SAT &\overset{f}{\Longrightarrow} \textsc{Block Sorting} \overset{g}{\Longrightarrow} \SSS,\\
    3SAT &\overset{f \circ g}{\Longrightarrow} \SSS,\\
    \noalign{and}
    3SAT &\leq_p \SSS.  
    \end{align*}
\end{center}
The reduction outlined above proves that \SSS{} is $\NP$ hard. \ignore{In Section~\ref{lb}, we present various lower bounds for \SSS. Section~\ref{redblue} discusses the construction of the red-blue graph of a permutation from a given block-sorting schedule. The red-blue graph would be instrumental in our reduction. In Section~\ref{construct} we present the construction of $\pi^{\prime}$, and obtain a strip-swap schedule $\mathcal{S}^{\prime}$ on $\pi^{\prime}$, given $\pi$ and a block sorting schedule $\mathcal{S}$ on $\pi$. Finally, in Section~\ref{npc} we establish the result that \SSS{} is $\NP$-Hard.}

\section{Lower Bounds for \SSS}\label{lb}
It is easy to see that $\SSS\in \NP$. 
\begin{property}\label{lb1}
In a single strip swap, the number of strips that could be reduced is at most $4$. Suppose that the number of strips in a permutation $\pi$ is denoted $\#strips$. Since $id$ contains $1$ strip, we have
\[ \SSD(\pi) \geq \left \lceil(\#strips-1)/4\right \rceil.\]
\ignore{
\[ \SSD(\pi) \geq \left \lceil\frac{\#strips-1}{4} \right \rceil.\]
}
\end{property}
\noindent
The following lemma provides another lower bound for \SSD.
\begin{lemma}\label{lb2}
$\SSD(\pi)\geq\left \lceil \rev(\pi)/2\right \rceil$, where  $\rev(\pi)$ is the number of reversals in $\pi$. 
\end{lemma}
  
\proof
We observe that a single strip swap can reduce $\rev(\pi)$ at most by $2$. Since $id$ does not contain any reversals, the goal of \SSS{} is to reduce $\rev(\pi)$ to $0$. Hence, it would require at least $\left \lceil \rev(\pi)/2 \right \rceil$ strip swaps to sort $\pi$.
\qed
\section{The Reduction}

\noindent
Let $\pi=a_1\ldots a_n$ and set $R=\mathrm{rev}(\pi)$. We build a permutation $\pi^\dagger$ and use the threshold $R$ for SbSS.

\subsection{Cage gadget (local isolation)}

\noindent
For each decreasing adjacency $(a_i,a_{i+1})$ output the block
\[
L_i\; a_i\; m_i\; a_{i+1}\; U_i
\]
with strict order constraints
$L_i < a_{i+1} < m_i < a_i < U_i$.
This creates exactly two internal decreases $(a_i,m_i)$ and $(m_i,a_{i+1})$ and guarantees that all adjacencies that cross the cage boundary are increasing (Fig.~\ref{fig:cage}).

\subsection{Hinge gadget (coupling shared elements)}

\noindent
If an element $a_j$ participates in two consecutive decreasing adjacencies $\big(a_{j-1},a_j\big)$ and $\big(a_j,a_{j+1}\big)$, we insert a pair $h_j^L,h_j^R$ positioned between the two cages so that:
\begin{itemize}[leftmargin=*,noitemsep]
  \item any strip swap that resolves the left cage in isolation leaves exactly one \emph{hinge penalty} (a single external decrease) unless the right cage is in the matching configuration, and
  \item symmetrically for resolving the right cage alone.
\end{itemize}

\noindent
The inequalities to realize this are straightforward (guards of the two cages bound a window in which $h_j^L,h_j^R$ sit strictly increasing to the outside); we omit them here for space (Fig.~\ref{fig:hinge}).

\subsection{Relabeling}

\noindent
Output increasing adjacencies of $\pi$ unchanged, concatenate all pieces, and relabel to a permutation on $\{1,\ldots,|\pi^\dagger|\}$ that respects the stipulated inequalities. The result has exactly $2R$ decreases (two per cage) and none elsewhere.

%\subsection{Illustrations}
\begin{figure}[!t]
  \centering
  \resizebox{\columnwidth}{!}{
  \begin{tikzpicture}[>=stealth]
    % Single cage
    \node[cell,fill=black!5] (L)  at (0,0) {$L_i$};
    \node[cell,fill=black!5,right=0.06 of L] (ai) {$a_i$};
    \node[cell,fill=yellow!20,right=0.06 of ai] (m) {$m_i$};
    \node[cell,fill=black!5,right=0.06 of m] (aip) {$a_{i+1}$};
    \node[cell,fill=black!5,right=0.06 of aip] (U) {$U_i$};
    % arrows marking internal decreases
    \draw[rev,->] (ai.south) .. controls +(0,-0.7) and +(0,-0.7) .. (m.south) node[midway,below=6pt,sloped]{\scriptsize dec};
    \draw[rev,->] (m.south) .. controls +(0,-0.7) and +(0,-0.7) .. (aip.south) node[midway,below=6pt,sloped]{\scriptsize dec};
  \end{tikzpicture}
  }
  \caption{\textbf{Cage gadget} for a reversal boundary $(a_i,a_{i+1})$: only internal decreases exist.}
  \label{fig:cage}
\end{figure}
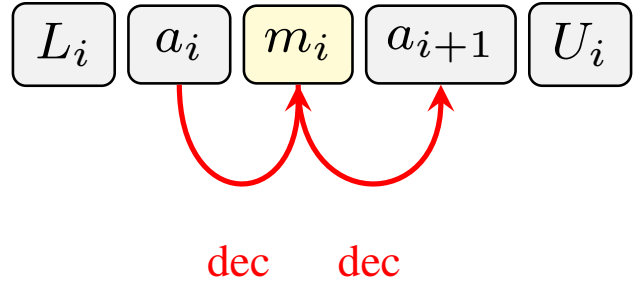

\begin{figure}[!t]
  \centering
  \resizebox{\columnwidth}{!}{
  \begin{tikzpicture}[>=stealth]
    % Left cage around (a_{j-1}, a_j)
    \node[cell,fill=black!5] (L1) at (0,0) {$L_{j-1}$};
    \node[cell,fill=black!5,right=0.06 of L1] (aL) {$a_{j-1}$};
    \node[cell,fill=yellow!20,right=0.06 of aL] (mL) {$m_{j-1}$};
    \node[cell,fill=black!5,right=0.06 of mL] (aJ) {$a_j$};
    \node[cell,fill=black!5,right=0.06 of aJ] (U1) {$U_{j-1}$};
    % hinge left token
    \node[cell,fill=blue!10,right=0.12 of U1] (hL) {$h_j^L$};
    % Right cage around (a_j, a_{j+1})
    \node[cell,fill=black!5,right=0.12 of hL] (L2) {$L_{j}$};
    \node[cell,fill=black!5,right=0.06 of L2] (aJ2) {$a_j$};
    \node[cell,fill=yellow!20,right=0.06 of aJ2] (mR) {$m_{j}$};
    \node[cell,fill=black!5,right=0.06 of mR] (aR) {$a_{j+1}$};
    \node[cell,fill=black!5,right=0.06 of aR] (U2) {$U_{j}$};
    % hinge right token
    \node[cell,fill=blue!10,right=0.12 of U2] (hR) {$h_j^R$};

    % internal decreases (two per cage)
    \draw[rev,->] (aL.south) .. controls +(0,-0.7) and +(0,-0.7) .. (mL.south);
    \draw[rev,->] (mL.south) .. controls +(0,-0.7) and +(0,-0.7) .. (aJ.south);
    \draw[rev,->] (aJ2.south) .. controls +(0,-0.7) and +(0,-0.7) .. (mR.south);
    \draw[rev,->] (mR.south) .. controls +(0,-0.7) and +(0,-0.7) .. (aR.south);

    % explanatory note for hinge penalty (conceptual)
    \node[below=10mm of aJ2,align=center,font=\scriptsize,draw=none] (note) {Hinge ensures: resolving only one cage\newline leaves a single external decrease};
  \end{tikzpicture}
  }
  \caption{\textbf{Hinge gadget}: two cages sharing $a_j$ coupled by $h_j^L,h_j^R$.}
  \label{fig:hinge}
\end{figure}
\section{Correctness}

\noindent
Let $\rho$ be the projection that deletes all guards and hinge tokens and contracts each cage to the boundary $(a_i,a_{i+1})$, producing a permutation over $\{a_1,\dots,a_n\}$.

\begin{lemma}[Forward (existence)]
If $bs(\pi)=R$, then $\mathrm{SSD}(\pi^\dagger)=R$.
\end{lemma}
\begin{proof}
A perfect block schedule has $R$ moves, each removing exactly one reversal boundary of $\pi$. For the corresponding cage in $\pi^\dagger$, swap the two singleton strips $[a_i]$ and $[a_{i+1}]$ inside the cage. By the guard inequalities, this eliminates the two internal decreases and cannot create new ones outside. Doing this for all $R$ cages sorts $\pi^\dagger$ in $R$ swaps, which is optimal by the SbSS lower bound.
\end{proof}

\begin{lemma}[Compatibility of $-2$ moves]
Any strip swap in $\pi^\dagger$ that reduces the number of decreasing adjacencies by exactly $2$ acts within a single cage by exchanging $[a_i]$ and $[a_{i+1}]$, and its projection $\rho$ is a block move in $\pi$ that reduces $\mathrm{rev}(\pi)$ by exactly $1$ and creates none elsewhere.
\end{lemma}
\begin{proof}
All decreases lie inside cages; guards keep the outside increasing. A $-2$ change must therefore resolve one cage. Hinges ensure that a $-2$ swap does not leave a hinge penalty, so its effect under $\rho$ is to flip exactly one decreasing adjacency in $\pi$ and no others.
\end{proof}

\begin{lemma}[Projection]
If $\pi^\dagger$ has an exact strip-swap schedule $S$ of length $R$, then $\rho(S)$ is a perfect block-sorting schedule of length $R$ on $\pi$.
\end{lemma}
\begin{proof}
Initially $\pi^\dagger$ has $2R$ decreases. Exactness in $R$ swaps forces every swap to be $-2$. By the previous lemma, each corresponds to a $-1$ block move under $\rho$. After $R$ swaps we have removed $R$ decreases in $\pi$. Hence $bs(\pi)\le R$, and by the lower bound we have equality.
\end{proof}

\begin{theorem}
$\pi$ has a perfect block-sorting schedule of length $R$ if and only if $\pi^\dagger$ has an exact strip-swap schedule of length $R$.
\end{theorem}
\begin{proof}
Immediate from the lemmas above.
\end{proof}

\begin{corollary}[NP-Hardness]
\textsc{Sorting by Strip Swaps} is NP-hard.
\end{corollary}
\begin{proof}
Reduce from \textsc{Block Sorting} with threshold $R=\mathrm{rev}(\pi)$. The construction is polynomial and preserves YES/NO by the theorem.
\end{proof}

\section{Worked Examples}
\paragraph*{YES instance.} $\pi_Y=4\,1\,3\,2$ has $R=2$ disjoint reversals. The constructed $\pi_Y^\dagger$ contains two independent cages; swapping $[4]\leftrightarrow[1]$ and then $[3]\leftrightarrow[2]$ sorts in $2=R$ swaps.

\begin{figure}[!t]
\centering
\resizebox{\columnwidth}{!}{%%
\begin{tikzpicture}[>=stealth]
% Row titles
\node[draw=none] at (-0.3,1.1) {\footnotesize Initial $\pi_Y^\dagger$};
% Initial state: two cages
\node[cell,fill=black!5] (L1) at (0,0) {$L_1$};
\node[cell,fill=black!5,right=0.06 of L1] (a1) {$4$};
\node[cell,fill=yellow!20,right=0.06 of a1] (m1) {$m_1$};
\node[cell,fill=black!5,right=0.06 of m1] (b1) {$1$};
\node[cell,fill=black!5,right=0.06 of b1] (U1) {$U_1$};
\node[cell,fill=black!5,right=0.18 of U1] (L3) {$L_3$};
\node[cell,fill=black!5,right=0.06 of L3] (a3) {$3$};
\node[cell,fill=yellow!20,right=0.06 of a3] (m3) {$m_3$};
\node[cell,fill=black!5,right=0.06 of m3] (b3) {$2$};
\node[cell,fill=black!5,right=0.06 of b3] (U3) {$U_3$};
% Internal decreases
\draw[rev,->] (a1.south) .. controls +(0,-0.7) and +(0,-0.7) .. (m1.south);
\draw[rev,->] (m1.south) .. controls +(0,-0.7) and +(0,-0.7) .. (b1.south);
\draw[rev,->] (a3.south) .. controls +(0,-0.7) and +(0,-0.7) .. (m3.south);
\draw[rev,->] (m3.south) .. controls +(0,-0.7) and +(0,-0.7) .. (b3.south);

% Arrow and label for swap1
\draw[->,thick] ([yshift=-0.95cm]a1.south) -- node[midway,fill=white,inner sep=1pt,font=\scriptsize]{swap 1: $[4]\leftrightarrow[1]$} ([yshift=-0.95cm]b1.south);

% After swap1 row
\node[draw=none] at (-0.7,-1.5) {\footnotesize after swap 1};
\node[cell,fill=black!5] (L1b) at (0,-2.1) {$L_1$};
\node[cell,fill=black!5,right=0.06 of L1b] (b1b) {$1$};
\node[cell,fill=yellow!20,right=0.06 of b1b] (m1b) {$m_1$};
\node[cell,fill=black!5,right=0.06 of m1b] (a1b) {$4$};
\node[cell,fill=black!5,right=0.06 of a1b] (U1b) {$U_1$};
\node[cell,fill=black!5,right=0.18 of U1b] (L3b) {$L_3$};
\node[cell,fill=black!5,right=0.06 of L3b] (a3b) {$3$};
\node[cell,fill=yellow!20,right=0.06 of a3b] (m3b) {$m_3$};
\node[cell,fill=black!5,right=0.06 of m3b] (b3b) {$2$};
\node[cell,fill=black!5,right=0.06 of b3b] (U3b) {$U_3$};
% Arrow and label for swap2
\draw[->,thick] ([yshift=-0.95cm]a3b.south) -- node[midway,fill=white,inner sep=1pt,font=\scriptsize]{swap 2: $[3]\leftrightarrow[2]$} ([yshift=-0.95cm]b3b.south);

% After swap2 row (sorted)
\node[draw=none] at (-0.7,-3.6) {\footnotesize after swap 2 (sorted)};
\node[cell,fill=black!5] (L1c) at (0,-4.2) {$L_1$};
\node[cell,fill=black!5,right=0.06 of L1c] (b1c) {$1$};
\node[cell,fill=yellow!20,right=0.06 of b1c] (m1c) {$m_1$};
\node[cell,fill=black!5,right=0.06 of m1c] (a1c) {$4$};
\node[cell,fill=black!5,right=0.06 of a1c] (U1c) {$U_1$};
\node[cell,fill=black!5,right=0.18 of U1c] (L3c) {$L_3$};
\node[cell,fill=black!5,right=0.06 of L3c] (b3c) {$2$};
\node[cell,fill=yellow!20,right=0.06 of b3c] (m3c) {$m_3$};
\node[cell,fill=black!5,right=0.06 of m3c] (a3c) {$3$};
\node[cell,fill=black!5,right=0.06 of a3c] (U3c) {$U_3$};
\end{tikzpicture}}
\caption{YES instance $\pi_Y=4\,1\,3\,2$: two independent cages; two local strip swaps (each $-2$) sort $\pi_Y^\dagger$ in $R=2$ moves.}
\label{fig:yes-example}
\end{figure}
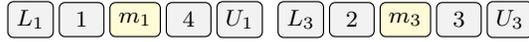

\paragraph*{NO instance.} $\pi_N=7\,2\,6\,5\,8\,3\,1\,4$ has $R=4$ with element $3$ shared by two reversals. In $\pi_N^\dagger$, the hinge between the two cages containing $3$ forces at least one non-$-2$ swap, so any schedule needs $\ge5$ swaps. Thus no exact $R$-swap schedule exists.

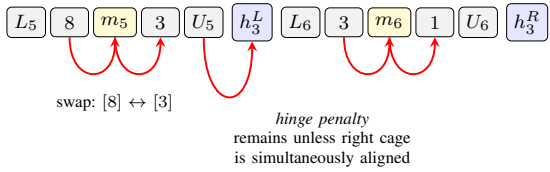
\begin{figure}[!t]
\centering
\resizebox{\columnwidth}{!}{%%
\begin{tikzpicture}[>=stealth]
% Initial layout with hinge between two cages around shared 3
\node[draw=none] at (-0.2,1.1) {\footnotesize Initial portion of $\pi_N^\dagger$ around $3$};
% Left cage (8,3)
\node[cell,fill=black!5] (L5) at (0,0) {$L_5$};
\node[cell,fill=black!5,right=0.06 of L5] (aL) {$8$};
\node[cell,fill=yellow!20,right=0.06 of aL] (mL) {$m_5$};
\node[cell,fill=black!5,right=0.06 of mL] (aJ) {$3$};
\node[cell,fill=black!5,right=0.06 of aJ] (U5) {$U_5$};
% hinge left token
\node[cell,fill=blue!10,right=0.12 of U5] (hL) {$h_3^L$};
% Right cage (3,1)
\node[cell,fill=black!5,right=0.12 of hL] (L6) {$L_6$};
\node[cell,fill=black!5,right=0.06 of L6] (aJ2) {$3$};
\node[cell,fill=yellow!20,right=0.06 of aJ2] (mR) {$m_6$};
\node[cell,fill=black!5,right=0.06 of mR] (aR) {$1$};
\node[cell,fill=black!5,right=0.06 of aR] (U6) {$U_6$};
% hinge right token
\node[cell,fill=blue!10,right=0.12 of U6] (hR) {$h_3^R$};

% Internal decreases
\draw[rev,->] (aL.south) .. controls +(0,-0.7) and +(0,-0.7) .. (mL.south);
\draw[rev,->] (mL.south) .. controls +(0,-0.7) and +(0,-0.7) .. (aJ.south);
\draw[rev,->] (aJ2.south) .. controls +(0,-0.7) and +(0,-0.7) .. (mR.south);
\draw[rev,->] (mR.south) .. controls +(0,-0.7) and +(0,-0.7) .. (aR.south);

% Indicate that resolving left cage alone leaves a hinge penalty
\draw[->,thick] ([yshift=-0.95cm]aL.south) -- node[midway,fill=white,inner sep=1pt,font=\scriptsize]{swap: $[8]\leftrightarrow[3]$} ([yshift=-0.95cm]aJ.south);
\node[draw=none,align=center,font=\scriptsize] at (4.3,-1.7) {\textit{hinge penalty}\\remains unless right cage\\is simultaneously aligned};
% Visualize a leftover single decrease at hinge (conceptual arc)
\draw[rev,->] (U5.south) .. controls +(0,-1.1) and +(0,-1.1) .. (hL.south);
\end{tikzpicture}}
\caption{NO instance $\pi_N=7\,2\,6\,5\,8\,3\,1\,4$: two cages share element $3$ and are coupled by a hinge. Any attempt to resolve one cage in isolation leaves a hinge penalty, so an exact $R$-swap schedule does not exist.}
\label{fig:no-example}
\end{figure}
\section{Conclusion}
We gave a schedule-free, local reduction from Block Sorting to SbSS. Cages isolate reversal boundaries; hinges couple neighbors so that exact strip swaps correspond one-for-one to perfect block moves. Composed with known NP-hardness of Block Sorting, this proves SbSS is NP-hard.

\balance

\bibliographystyle{IEEEtran}
% Replace the keys below with your actual entries
\bibliography{ref}
% Example placeholder entries in refs.bib:
% @inproceedings{refBlockSorting, title={Block Sorting is NP-hard}, author={Christie, D. A.}, booktitle={Proc. SODA}, year={1996}}
% @article{refStripSwaps, title={On an Algorithm for Sorting by Strip Swaps Using Cycle Graphs}, author={<authors>}, journal={<venue>}, year={<year>}}

\end{document}